\documentclass[conference]{IEEEtran}
\IEEEoverridecommandlockouts
\usepackage{cite}
\usepackage[belowskip=-5pt,aboveskip=0pt]{caption}
\usepackage{amsmath,amssymb,amsfonts}
\usepackage{amsthm}
\usepackage{algorithm2e}
\usepackage{graphicx}
\usepackage{textcomp}
\usepackage{xcolor}
\def\BibTeX{{\rm B\kern-.05em{\sc i\kern-.025em b}\kern-.08em
    T\kern-.1667em\lower.7ex\hbox{E}\kern-.125emX}}
\usepackage{etoolbox}
\usepackage{multirow}
\usepackage{tabularx}
\patchcmd{\thebibliography}{\section*{\refname}}{}{}{}
\newcommand{\edge}[1]{\langle  #1 \rangle}
\newtheorem{theorem}{Theorem}
\newtheorem{lemma}[theorem]{Lemma}
\newtheorem{definition}{Definition}

\newtheorem{innercustomthm}{Theorem}

\usepackage{mathtools}
\newcommand\myassignment{\stackrel{\mathclap{\normalfont\mbox{$>$}}}{\longleftarrow}}
\DeclareMathOperator*{\argmax}{arg\,max}
\usepackage{algorithmicx}
\usepackage{algpseudocode}
\algdef{SE}[DOWHILE]{Do}{doWhile}{\algorithmicdo}[1]{\algorithmicwhile\ #1}%

\begin{document}

\title{Contour Algorithm for Connectivity}

\author{\IEEEauthorblockN{Zhihui Du, Oliver Alvarado Rodriguez,  Fuhuan Li, Mohammad Dindoost and David A. Bader}
\IEEEauthorblockA{\textit{Department of Data Science} \\
\textit{New Jersey Institute of Technology}\\
Newark, USA \\
\{zhihui.du,oaa9,fl28,md724,bader\}@njit.edu}
}

\maketitle

\begin{abstract}
Finding connected components in a graph is a fundamental problem in graph analysis. In this work, we present a novel minimum-mapping based \emph{Contour} algorithm to efficiently solve the connectivity problem. We prove that  the \emph{Contour} algorithm with two or higher order operators can identify all connected components of an undirected graph within $\mathcal{O}(\log d_{max})$ iterations, with each iteration involving $\mathcal{O}(m)$ work, where $d_{max}$ represents the largest diameter among all components in the given graph, and $m$ is the total number of edges in the graph. Importantly, each iteration is highly parallelizable, making use of the efficient minimum-mapping operator applied to all edges. To further enhance its practical performance, we optimize the \emph{Contour} algorithm through asynchronous updates, early convergence checking, eliminating atomic operations, and choosing more efficient mapping operators. 
Our implementation of the \emph{Contour} algorithm has been integrated into the open-source framework Arachne. Arachne extends Arkouda for large-scale interactive graph analytics, providing a Python API powered by the high-productivity parallel language Chapel. Experimental results on both real-world and synthetic graphs demonstrate the superior performance of our proposed \emph{Contour} algorithm compared to state-of-the-art large-scale parallel algorithm FastSV and the fastest shared memory algorithm ConnectIt. On average, \emph{Contour} achieves a speedup of 7.3x and 1.4x compared to FastSV and ConnectIt, respectively. All code for the \emph{Contour} algorithm and the Arachne framework is publicly available on GitHub~\footnote{https://github.com/Bears-R-Us/arkouda-njit}, ensuring transparency and reproducibility of our work.
\end{abstract}

\begin{IEEEkeywords}
connected components, graph analytics, big data, parallel algorithm
\end{IEEEkeywords}

\section{Introduction}
A graph is one of the fundamental mathematical structures used to model pairwise relations between abstract objects. Many problems in science, society, and economics can be modeled by graphs. The sizes of graph data collections continue to grow which makes the need for fast graph algorithms critical, especially under online and real-time scenarios.

Finding connected components \cite{blelloch2020parallelism,karger1992fast,grana2010optimized,chong1995finding}
is a fundamental problem in graph analytics and an important first step for other graph algorithms. Many graph algorithms are based on the assumption that we already know a graph's connected components. In this work, we focus on the connectivity of undirected graphs. The connected components problem can be expressed as assigning each vertex with a label. If two vertices are in the same component or there is a path between them, they will be marked with the same label. Otherwise, the vertices will be marked with different labels \cite{cormen2022introduction}.

There are three kinds of algorithms for identifying connected components of an undirected graph. The first is a graph traversal-based method. Breadth-First Search (\emph{BFS}) \cite{jain2017adaptive} and label propagation \cite{esfahani2021thrifty,stergiou2018shortcutting,raghavan2007near} are two typical examples. \emph{BFS} will search from a set of just visited vertices (current frontier, initially with one root vertex) and then extend to other unvisited vertices (next frontier) connected to visited vertices until all vertices are visited.
The basic idea of label propagation is that each vertex is initially assigned a unique label. Then, each vertex subsequently compares its label with the labels of its neighbors and updates its label to be the smallest among them. This process is repeated 
until no label can be updated.
There are many variants to improve the performance further. This method has high performance for low-diameter graphs. However, if a graph has a long diameter, a lot of time and iterations will be needed.

The second is the tree hooking-compressing-based method \cite{SHILOACH198257, krishnamurthy1997connected,azad2019lacc,zhang2020fastsv}. This kind of method will start by initializing all vertices as singletons. Then, some tree hooking operations are employed to merge smaller components into larger components. Compressing operations will reduce the tree's height until all vertices are directly connected to a root vertex. The major feature of such a method is formulating the discovery of a big component as a forest building.  Combining the tree hooking and compressing, a much smaller number of iterations will be needed to identify all the components, even if the given graph has a large diameter. 

The third is the union-find or disjoint set-based method \cite{dhulipala2020connectit,galil1991data,patwary2010experiments}. It models components as disjoint sets. The union operation will merge different sets and the find operation will return the representative member of a set. Unlike the previous two methods, union-find is not an iteration-based method. It can directly identify all connected components in one iteration of the tree-based method. However, for large-scale parallelism scenarios, union-find methods often lead to an unbalanced workload that can significantly affect their performance. 

We abstract the connectivity as a contour lines discovery problem and develop simple and lightweight minimum-mapping operators to work on different edges to efficiently identify all the components in parallel. The minimum-mapping operator can map the connected vertices to the same contour line. Identifying one component is similar to identifying one contour line with the same minimum mapping label. Therefore, we name our algorithm ``\emph{Contour}"\cite{courant1996mathematics}. The minimum-mapping operator can be employed on different edges in parallel with high efficiency. Compared to tree hooking-compressing or union-find-based methods, this can significantly improve parallel performance and simplify implementation. 

The major contributions of this work are as follows.
\begin{enumerate}
    \item A novel \emph{Contour} algorithm that formulates finding connected components as a contour lines discovery problem. Based on this perspective, simple and lightweight minimum-mapping operators are developed to map the vertices in the same component to the same label in parallel. The proposed method is suitable for large graphs with different graph topologies. 
    \item A proof is given to show that for a graph with $d_{max}$ as its maximum diameter, the \emph{Contour} algorithm can converge in $\mathcal{O}(log (d_{max}))$ iterations. 
    \item The proposed method has been integrated into the graph package, Arachne. It is publicly available through the open-source Arkouda framework from GitHub to analyze large graphs using the popular Python interface.
    \item Extensive experimental results show that the proposed \emph{Contour} algorithm can achieve significant speedup compared to state-of-the-art real-world and synthetic graphs methods.
\end{enumerate}

\section{Contour Algorithm}
\subsection{Problem Description}
Given an undirected graph $G=<V,E>$, where $V$ is the set of vertices, and $E$ is the set of edges. Let $m=|E|$ be the total number of edges and $n=|V|$ be the total number of vertices in $G$. Without loss of generality, here we assume that vertex IDs are from $0$ to $n-1$.  

A label array $L[0..n-1]$ with size $n$ can be used to store all the labels of different vertices. Initially, we assign each vertex's ID as its label. Identifying all connected components in $G$ means that we will assign the vertices of the same components with the same vertex label. The label array is also regarded as a pointer graph \cite{SHILOACH198257}. $\forall v \in V, L[i]=v$ means that there is a direct edge from vertex $i$ to $v$. The pointer graph will be updated after each iteration. It is a forest of rooted trees plus self-loops that occur only in the root.  Finally, if graph $G$ has $S$ components, L will represent $S$ stars after all components are found. A star here is a unique type of graph characterized by a single root vertex connected to all other vertices, with no additional edges present.

\subsection{Minimum-Mapping Operator} \label{subsec:operator}
 $\forall v \in V$, $L[v]$ is the mapped vertex or label of $v$. $L_u[0..n-1]$ is used to store the updated value of different vertices after once iteration. If there is a path between $w$ and $v$ or $w$ and $v$ are connected, and the values of their labels are different, we should assign them the same label. Here we use the minimum value among $L[w]$ and $L[v]$ to update the old label values in $L_u$ array.

First, we define the conditional vector assignment operator as follows.
\begin{definition}[Conditional Vector Assignment]\label{definition:assignment}
 \begin{align}
     \begin{bmatrix}
           x_1 \\
           ...\\
           x_k \\
     \end{bmatrix}
         \myassignment z.
  \end{align}
It means that given a vector $X=$ 
$\begin{bmatrix}
           x_1 \\
           ...\\
           x_k \\
\end{bmatrix}$ $, \forall i \in \mathbb{N}, 1 \le i \le k$,  $x_i=z$ if $x_i >z$.

\end{definition}

Based on the definition of conditional vector assignment, we will further define our minimum-mapping operators. 
\begin{definition}[One-Order Minimum-Mapping Operator]
Given two connected vertices $w,v \in V$, let $z^1=min(L[w],L[v])$. We define the one-order minimum-mapping operator as follows.
 \begin{align}
    MM^1(L_u, L,w,v): & \begin{bmatrix}
           L_u[w] \\
           L_u[v] \\
         \end{bmatrix}
         \myassignment 
           z^1
  \end{align}
  \end{definition}
$MM^1(L_u,L,w,v)$ means that before the mapping operator, $L_u=L$. After employing the mapping operator, $L_u[w]$ and  $L_u[v]$ will be updated if either of them is larger than $z^1$. 

Higher \emph{h$>$1} order minimum-mapping operators $MM^h(L_u,L,w,v)$ can also be defined similarly.
\begin{definition}[$h$-Order Minimum-Mapping Operator]
Given two connected vertices $w,v \in V$, let  $z^h=min(L^h[w],L^h[v])$, where $\forall x \in V, L^{h}[x]=L[L^{h-1}[x]], L^1[x]=L[x]$. We define the h-order minimum-mapping operator as follows. 
\begin{align}
    MM^h(L_u,L,w,v): & \begin{bmatrix}
           L_u[w] \\
           L_u[v] \\
           ...\\
           L_u[L^{h-1}[w]]\\
           L_u[L^{h-1}[v]] \\
         \end{bmatrix}
         \myassignment  z^h.
\end{align}
\end{definition}

A higher-order minimum-mapping operator may include more mapped vertices based on the two given vertices. So it may find the final minimum contour quickly. However, it will also perform many more operations. In this paper, we take the two-order minimum-mapping operator as the default operator because it can achieve a quick convergence (logarithmic time complexity) with a minimum-mapping operator involving a much smaller number of vertices and operations. We will also show the different effects of its variants and combination patterns in Section \ref{sec:experients}. 

\subsection{Algorithm Description}
Based on the proposed minimum-mapping operator in Section~\ref{subsec:operator}, our \emph{Contour} algorithm is given in Alg. \ref{alg:connectedcomponents}. The complete algorithm is straightforward and easy to parallelize. 

For lines from 1 to 4, we initialize the label array $L$ and the corresponding update array $L_u$ with each vertex's ID. From line 5 to line 10, we update the label array $L$ until convergence or there are no changes in the array. From lines 6 to 8, for each edge $e=\edge{w,v} \in E$, we will execute the two-order minimum-mapping $MM^2(w,v)$ in parallel. $MM^2(w,v)$ may update the value of $L_u[w], L_u[v], L_u[L[w]], L_u[L[v]]$ if they are larger than the minimum value $z^2$. In line 9, all the old values in $L$ will be updated with the new values in $L_u$.

Since all the conditional assignments can be executed in parallel, to avoid write races, we can use the atomic compare-and-swap (CAS) \footnote{https://chapel-lang.org/docs/primers/atomics.html} operation to implement our conditional assignment as follows. 
\begin{equation} \label{eq:cas}
\begin{aligned}
&while \;(oldx_i=atomic\_read(x_i) > z) \; \{\\
&\qquad \;CAS(x_i,oldx_i,z)&\\
&\}
\end{aligned}
\end{equation}

\RestyleAlgo{ruled}
\begin{algorithm}[tpbh]
\small
\DontPrintSemicolon
\LinesNumbered
\caption{Minimum-Mapping based \emph{Contour} Algorithm}
\label{alg:connectedcomponents}
\SetKwRepeat{Do}{do}{while}
$Contour(G)$\\
\tcc{$G=\edge{E,V}$ is the input graph with edge set $E$ and vertex set $V$.}
\ForAll {i in 0..n-1} {
    $L[i]=i$\;
    $L_u[i]=i$\;
}
\tcc{Initialize the label array $L,L_u$}
\Do {(There is any label change in $L$)}{
    \ForAll  {($e=\edge{w,v} \in E$)} {
        $MM^2(L_u,L,w,v)$ \;
    }
    $L=L_u$\;
}
\Return $L$ \;
\end{algorithm}
Let $L_k^h[x]$ be the label of vertex $x$ employing the $h$-order minimum-mapping operator after the $k^{th}$ iteration, then  $L_{u,k}[w]=min(L_{k-1}^2[w],L_{k-1}^2[v_1],L_{k-1}^2[v_2],...,$ $L_{k-1}^2[v_m]$, where   $v_1,v_2,...,v_m$ are the vertices that directly connect with $w$, or directly connect with the vertices that are mapped to $w$.  
We give the following definition to show how the vertices in the same component are mapped to the same minimum label step by step. 
\begin{definition}[Equal Minimum Set]
Given label $x$, after the $k^{th}$ iteration, its one-order equal minimum set $EMS(k)_x^1=\{v|\forall v \in V, L_k[v]=x\}$.  Its two-order equal minimum set $EMS(k)_x^2=\{v|\forall v \in V, L_k^2[v]=x\}$.  
\end{definition}
We use the equal minimum set to indicate the vertices mapped to the same vertex label. 
\begin{definition}[Merged Minimum Set]
Let  $MMS(0)=V$. After the $k^{th}$ iteration, $k\geq 1$, the one-order merged minimum set is defined as $MMS(k)^1=\{v|\forall v \in V, EMS(k)_v^1 \neq \phi \}$. Similarly, the two-order merged minimum set $MMS(k)^2=\{v|\forall v \in V, EMS(k)_v^2 \neq \phi \}$.  
\end{definition}
From the definition, we can see that for $k \ge 0, MMS(k)^1 \supseteq MMS(k)^2 \supseteq MMS(k+1)^1 $. In other words, the merged minimum set's size will become smaller until it only contains the minimum vertices of different connected components. 

\begin{definition}[Rooted Tree and its Neighbor]
After the $k^{th}$ iteration, $k\ge 1$, the root vertices of different root trees in the pointer graph $R(t)=\{v|v \in MMS(k) \wedge L_k[v]=v\}$. $R(t)$ is also called the root tree set of the pointer graph. $\forall v_1,v_2 \in R(t)$, if $\exists \edge{v_{1'},v_{2'}} \in E$,  and $v_{1'}$ belongs to in root tree $v_1$, $v_{2'}$ belongs to root tree $v_2$, then we call $v_1$ the neighbor of $v_2$, vice versa. 
\end{definition}
Our mapping operator has the following two effects on the rooted trees. (1) Compressing. If the original height of a rooted tree is $x$ and we employ $h$ order minimum mapping operator to it, its height will be reduced to no more than $\lfloor \frac{x+h-1}{h} \rfloor$. Every vertex in the rooted tree will point to its $h$ order father or the root. (2) Minimum Merging. Any vertex $v_m$ in one rooted tree may be merged into its neighbor rooted tree as the son of root or other vertices. At the same time, the subtree (if exists) with $v_m$ as its root will be compressed and merged into its neighbor-rooted tree. Both compressing and merging can happen at the same time in one minimum-mapping operation. Minimum merging is very flexible and different from the existing tree-hooking or set union methods. One rooted tree can merge part of another rooted tree instead of the complete rooted tree. At the same time, it is simple and easy to implement. 

 The framework of Alg. \ref{alg:connectedcomponents} has some similarities to label propagation or tree hooking-compressing. However, the label propagation method can be regarded as a special case of our method when the mapping order is one. Compared with the existing tree hooking-compressing methods, they only allow merging two rooted trees. However, our method can merge any part of two rooted trees.  
 
 The following section will prove that our Alg.\ref{alg:connectedcomponents} can converge in logarithmic iterations. 

\subsection{Time complexity analysis}

\begin{lemma}[Root Tree Constraint] \label{lemma:tree}
Let $P = \edge{s_0, \ldots, s_{n-1}}, n\ge 2$, be a path with $s_0$ as the smallest vertex, and consider running Alg. \ref{alg:connectedcomponents} on $P$ (here we assume the mapping operator can be employed up to twice in each iteration). After the $k^{th}$ iteration, let the root tree set be $R(k)$,  we have $(\frac{3}{2})^{k-1}\sum_{v \in R(k)}H_k(v)\le n-1$, where $H_k(v)$ is the height of root tree $v$ after the $k^{th}$ iteration. 
\end{lemma}
\begin{proof}
Let's do induction on $k$.

For $k=1$, if all vertices are in increasing order along the path $P$, then $MMS(1)^1=s_0$ and $H_1(s_0)=n-1$, $\sum_{v\in R(1)}H_1(v) = n-1$. So, the inequation holds. 

Otherwise, if there are multiple root trees $v_{m_1},...,v_{m_j}$ in the pointer graph and we let $n_s$ be the number of vertices in root tree $v_{m_s}$, where $1\le s \le j$. So, we have $\sum_{s=1}^{s=j} n_s=n$ and $H_1(v_{m_s}) \le n_s-1$. Therefore, $\sum_{s=1}^{s=j}H_1(v_{m_s})\le n-j <n-1$. So, the inequality holds for the base case.

Let $t\ge 1$; we assume that when $t=k$, the inequation holds. Now we prove when $k=t+1$, the inequation also holds. If $R(t)=\{s_0\}$ and $H_t(s_0)>1$, then after the $(t+1)$ iteration, $H_{t+1}(s_0)\le(\frac{2}{3})H_t(s_0)$, so the inequation holds. 

We discuss two cases if $|R(t)|>1$.

(1) If  $\forall v \in R(t)$, $H_t(v)>1 \wedge v \in R(t+1)$, then  $H_{t+1}(v)\le(\frac{2}{3})H_t(v)$. 

(2) If $\exists v\in R(t) \wedge H_t(v)=1 \wedge L_{k+1}[v]=v'\wedge v\ne v'$, then root tree $v$ will be merged into the root tree $v_m$ that contains vertex $v'$ after the $(t+1)^{th}$ iteration. If $H_k(v_m)=1$, then $H_{k+1}(v_m)=1 < \frac{2}{3}(H_k(v)+H_k(v_m))=\frac{2}{3}\times 2=\frac{4}{3}$. If $H_k(v_m)>1$, we know that $H_{t+1}(v_m)\le \frac{2}{3}H_t(v_m)< \frac{2}{3}(H_t(v_m)+H_t(v))$. 

If $\exists v\in R(t) \wedge H_t(v)=1 \wedge L_{k+1}[v]=v$, it means that $v$ is less than its neighbour vertex $v_x \in R(t)$. So, $v$ will merge its neighbor root tree or partial vertices of its neighbor root tree. Since merging the complete root tree is the same as in the above case, we only consider the case when only partial vertices are merged into $v$. In this case, the neighbor root tree $v_n$ must have $H_t(v_n)>1$. Otherwise, the neighbor root tree will be merged into the $v$ root tree. Here, we can employ the mapping operator twice. If  $H_{t+1}(v_n)=1$, then $H_{t+1}(v_n)+H_{t+1}(v)=2$. We have $ \frac{2}{3} (H_t(v_n)+H_t(v_n)) \ge \frac{2}{3} (1+2)=2$.   If  $H_{t+1}(v_n)>1$, then we have $H_{t+1}(v_n)\le \frac{1}{3}H_{t}(v_n)$ so $H_{t+1}(v_n)+1\le \frac{2}{3}(H_{t}(v_n)+1)$ when $H_{t+1}(v_n)> 1$.

Hence, considering all the cases, we also have the same conclusion.
\end{proof}

\begin{lemma}[Path Convergence] \label{lemma:path}
Let $P = \edge{s_0, \ldots, s_{n-1}}, n\ge 2$, be a path with $s_0$ as the smallest vertex, and consider running Alg. \ref{alg:connectedcomponents} on $P$. Marking all vertices as $s_0$ will need at most $\lceil \log_{\frac{3}{2}}(n-1)\rceil+1$ iterations.
\end{lemma}

\begin{proof}
Based on Lemma \ref{lemma:tree}, 
$k$ $\le$ $ \lceil \log_{\frac{3}{2}} {\frac{(n-1)}{\sum_{v \in R(k)}H_k(v)}} $$\rceil+1$, when $\sum_{v \in R(k)}H_k(v)=1$, the maximum value of $k$ should be $\lceil \log_{\frac{3}{2}} (n-1)\rceil +1$. So, after at most $\lceil \log_{\frac{3}{2}} (n-1) \rceil+1$ iterations, all vertices on $P$ will be  marked as $s_0$. 
\end{proof}

\begin{lemma}[Diameter Convergence] \label{lemma:diameter}
For a connected graph $G$ with diameter $d$, Alg. \ref{alg:connectedcomponents} will take at most $(\lceil \log_{\frac{3}{2}}(d)\rceil+1)$ iterations to spread the minimum vertex label to all the other vertices.  
\end{lemma}

\begin{proof}
Let the smallest vertex in $G$ be $s_0$, then all shortest paths from $s_0$ to other vertices cannot be larger than $d$. Based on Lemma \ref{lemma:path}, the vertices on any shortest path from $s_0$ to other vertices can be mapped to $s_0$ within $(\lceil \log_{\frac{3}{2}}(d)\rceil+1)$ iterations. So the conclusion holds.
\end{proof}

\begin{innercustomthm}[Graph Convergence]\label{theorem:converge}
For any graph $G$, let $d_{max}$ be the maximum diameter of all graph $G$'s components. Alg. \ref{alg:connectedcomponents} will take at most $(\lceil \log_{\frac{3}{2}}(d_{max})\rceil+1)$ iterations to identify all the components.
\end{innercustomthm}

\begin{proof}
Let $G_c$ be any connected component of $G$ and $d$ be its diameter with $d\le d_{max}$. Based on Lemma \ref{lemma:diameter}, we know that after $(\lceil \log_{\frac{3}{2}}(d)\rceil+1)$ iterations, all vertices in $G_c$ will be mapped to their minimum vertices. Since $d_{max}$ is the maximum diameter of all graph $G$'s components, after $\lceil \log_{\frac{3}{2}}(d_{max})\rceil+1)$ iterations, all connected components of $G$ must have been mapped to their minimum vertices. So, Alg. \ref{alg:connectedcomponents} will take at most $\lceil \log_{\frac{3}{2}}(d_{max})\rceil+1)$ iterations to identify all the components.
\end{proof}

\section{Integration with Arachne and Performance Optimization}
\subsection{Integration Method}
Our method is integrated into Arachne \cite{2022-RDPLB}, a large-scale graph analytics package on top of Arkouda \cite{merrill2019arkouda,reus2020chiuw}. Arkouda is an open-source framework in Python created to be a NumPy replacement at scale. It replaces the \textit{ndarray} abstraction with the \textit{pdarray}. Our work aims to extend Arkouda for graph analytics, where we use the underlying \textit{pdarray} to implement and execute our algorithms. Through this, we create an end-to-end response system from Chapel to Arkouda. In Python, our calling method is called \textit{graph\_cc(graph)} where the user passes to a function a graph. We added this method to Arkouda's front-end file called \textit{graph.py}. The calling messages are added into \textit{arkouda\_server.chpl}. The Chapel method is invoked when the function is called in Python, and the messages are passed from Python to Chapel through ZMQ \footnote{https://zeromq.org/}. The messages are recognized at the back-end by \textit{arkouda\_server.chpl}, and the proper functions are invoked and executed in the chapel back-end.

\subsection{Algorithm Optimization}

Alg. \ref{alg:connectedcomponents} presents the fundamental concept of our method. However, we can further optimize it to enhance its practical performance when we integrate the method into Arachne.

\subsubsection{Asynchronous Update} The essence of the asynchronous \emph{Contour} algorithm is to update the label array $L$ immediately, eliminating the need for maintaining an update label array $L_u$. An asynchronous update will not affect the correctness or final convergence of the algorithm. However, the practical performance will be very different. This approach offers several advantages:

(1) Faster convergence speed: Vertices can be mapped to lower labels more rapidly.

(2) Reduction of unnecessary operations: The step $L=L_u$ in Alg. \ref{alg:connectedcomponents} becomes unnecessary and can be removed.

(3) Memory usage reduction: The $L_u$ array is no longer required and can be eliminated.

Experimental results in Section \ref{sec:experients} demonstrate that asynchronous updates significantly improve the performance of the algorithm.

\subsubsection{Early Convergence Check}

With the definition of our minimum mapping operator, for any edge $e=(v,w) \in E$, if $(L[v] \neq L^2[v] || L[w] \neq L^2[w] || L[v] \neq L[w])$, we need to continue to the next iteration. However, if these conditions are not met, even if there are updates in the current iteration, we can confidently conclude that the algorithm has converged, and we can exit the iteration directly. This early convergence check allows us to save additional iterations.

By performing this convergence check, we can efficiently terminate the algorithm once the convergence condition is met, reducing unnecessary computations and improving the overall efficiency of the algorithm.

\subsubsection{Eliminating Atomic Operations}

In union-find algorithms, atomic operations are essential to ensure correctness. However, in iteration-based methods, atomic operations can impact the number of iterations but not the correctness of the algorithm. These atomic operations can be computationally expensive compared to simple assignments.

Utilizing asynchronous updates can accelerate the convergence speed and reduce the total number of iterations. Consequently, we have the opportunity to replace costly atomic updates with simple assignments, further enhancing the practical performance of the algorithm.

By removing atomic operations and employing simple assignments, we can achieve better computational efficiency without compromising the correctness of the algorithm. It is similar to the effect of replacing synchronization updates with asynchronous updates. These optimizations contribute to the overall improvement in practical performance.

\subsubsection{Selecting Suitable Minimum Mapping Operators}

The choice of minimum mapping operators and their combination patterns can also significantly impact the performance of the algorithm for a given graph. We will provide recommendations based on experimental results in subsection \ref{subsec:speedup1}. Here, we introduce six different variants of our \emph{Contour} algorithm:

\emph{C-Syn}: This is the synchronous method described in Alg. \ref{alg:connectedcomponents} without employing any other optimization methods. Except for the minimum mapping operator, it is almost the same as the \emph{FastSV} algorithm. It can only achieve limited speedup compared with \emph{FastSV}.

\emph{C-1}: This variant employs the one-order minimum mapping operator.

\emph{C-2}: This variant employs the two-order minimum mapping operator.

\emph{C-m}: For large-diameter graphs, we may use a higher-order minimum mapping operator greater than two to reduce the total number of iterations.

\emph{C-11mm}: This variant combines operators. It starts with the one-order mapping operator for a few iterations and then switches to a higher-order operator until convergence.

\emph{C-1m1m}: This variant alternates between the one-order and higher-order operators until convergence.

In subsection \ref{subsec:speedup1}, we will provide specific analysis and guidance on selecting the most suitable variant based on the characteristics of the graph to achieve optimal performance.

\subsection{State-of-the-Art Algorithms}

In addition to our \emph{Contour} algorithm, we have incorporated two state-of-the-art algorithms, namely \emph{FastSV} \cite{zhang2020fastsv} and \emph{ConnectIt} \cite{dhulipala2020connectit}, into Arachne. This integration allows us to expand our algorithm repository and conduct performance comparisons.

The seminal Shiloach-Vishkin (SV) algorithm \cite{SHILOACH198257} is capable of identifying graph components in $\mathcal{O}(\log(n))$ time on a CRCW PRAM machine with $(m+n)$ parallel processors. Various variants and improvements \cite{1302951} based on the original SV algorithm exist, and \emph{FastSV} represents the latest and most advanced version for large-scale parallel systems. However, the synchronization introduced between different hooking and compressing steps, along with the need to assign the current label array $L$ with the updated label array $L_u$ before the next iteration, significantly affects its performance compared to our simple and flexible minimum mapping operators.

Union-find algorithms were developed to handle disjoint set data structures and can achieve almost linear time complexity \cite{tarjan1984worst}. Patwary \emph{et al.}'s \cite{patwary2010experiments} experimental results reveal that Rem's simple union-find algorithm delivers superior practical performance. Dhulipala \emph{et al.} developed the \emph{ConnectIt} framework, which incorporates hundreds of different connected components algorithms, and their experimental results on large graphs demonstrate that Rem's algorithm is the best within their shared memory system. In Arachne, we have integrated the optimal union-find algorithm from the ConnectIt framework for comparison. Our experimental results (see subsection \ref{subsec:speedup2}) demonstrate that our \emph{Contour} algorithm can effectively exploit parallel resources to achieve improved performance.

\section{Experiments} \label{sec:experients}

\subsection{Dataset Description}

Our dataset comprises a selection of publicly available synthetic and real-world datasets, sourced from reputable repositories. We have drawn graphs from the SuiteSparse Matrix Collection\footnote{https://sparse.tamu.edu/}, Stanford Large Network Dataset Collection\footnote{https://snap.stanford.edu/data/}, and the MIT GraphChallenge graph datasets\footnote{https://graphchallenge.mit.edu/data-sets}.

To ensure comprehensive testing and performance comparison, we have carefully chosen a combination of real-world and synthetic graphs, as outlined in Table \ref{tab:dataset2}. Including both types of graphs allows us to highlight and evaluate the efficacy of our \emph{Contour} algorithm compared with state-of-the-art algorithms.

The real-world graphs in our dataset may vary in characteristics. They typically exhibit degree distributions that follow a power-law distribution. These features provide a diverse set of challenges and scenarios to thoroughly assess the algorithms' performance.

Additionally, we have included a set of synthetic graphs known as \emph{Delaunay}. These graphs are constructed based on Delaunay triangulations of randomly generated points in the plane. Unlike graphs with power law distribution, \emph{Delaunay} graphs have vertices with degrees that are relatively close to each other. Including synthetic graphs allows us to observe how the algorithm's performance varies with graph size.

By incorporating a diverse range of datasets, our evaluation encompasses various real-world scenarios and provides valuable insights into the scalability and effectiveness of our \emph{Contour} algorithm, as well as its comparison with state-of-the-art methods.

\newsavebox{\boxDataSetb}
\begin{lrbox}{\boxDataSetb}
\begin{tabular}{|c|c|c|c|}
\hline
Graph Name&Graph ID&Number of Edges&Number of Vertices\\ \hline
ca-GrQc&0&28980&5242\\ \hline
ca-HepTh&1&51971&9877\\ \hline
facebook\_combined&2&88234&4039\\ \hline
wiki&3&103689&8277\\ \hline
as-caida20071105&4&106762&26475\\ \hline
ca-CondMat&5&186936&23133\\ \hline
ca-HepPh&6&237010&12008\\ \hline
email-Enron&7&367662&36692\\ \hline
ca-AstroPh&8&396160&18772\\ \hline
loc-brightkite\_edges&9&428156&58228\\ \hline
soc-Epinions1&10&508837&75879\\ \hline
com-dblp&11&1049866&317080\\ \hline
com-youtube&12&2987624&1134890\\ \hline
amazon0601&13&2443408&403394\\ \hline
soc-LiveJournal1&14&68993773&4847571\\ \hline
higgs-social\_network&15&14855842&456626\\ \hline
com-orkut&16&117185083&3072441\\ \hline
road\_usa&17&28854312&23947347\\ \hline
kmer\_A2a&18&180292586&170728175\\ \hline
kmer\_V1r&19&232705452&214005017\\ \hline
uk\_2002&20&298113762&18520486\\ \hline
delaunay\_n10&21&3056&1024\\ \hline
delaunay\_n11&22&6127&2048\\ \hline
delaunay\_n12&23&12264&4096\\ \hline
delaunay\_n13&24&24547&8192\\ \hline
delaunay\_n14&25&49122&16384\\ \hline
delaunay\_n15&26&98274&32768\\ \hline
delaunay\_n16&27&196575&65536\\ \hline
delaunay\_n17&28&393176&131072\\ \hline
delaunay\_n18&29&786396&262144\\ \hline
delaunay\_n19&30&1572823&524288\\ \hline
delaunay\_n20&31&3145686&1048576\\ \hline
delaunay\_n21&32&6291408&2097152\\ \hline
delaunay\_n22&33&12582869&4194304\\ \hline
delaunay\_n23&34&25165784&8388608\\ \hline
delaunay\_n24&35&50331601&16777216\\ \hline
\end{tabular}
\end{lrbox}

\begin{table}[thbp]
\centering
\caption{Real World and  Synthetic graphs}
\scalebox{0.65}{\usebox{\boxDataSetb}}
\label{tab:dataset2}
\end{table}

\subsection{Experimental Platform}\label{subsection:platform}
Experiments were done on a 32-node cluster system. Each node is a CentOS Linux release 7.9.2009 (Core) high-performance server with 2 x Intel Xeon E5-2650 v3 @ 2.30GHz CPUs with ten cores per CPU. Each server has 512GB of RAM. A high-performance Infiniband network system connects all nodes.

\subsection{Number of Iterations} \label{subsection:iteration}
\begin{figure*}
    \centering
    \hspace*{-2.6cm}
    \includegraphics[width=1.29\linewidth]{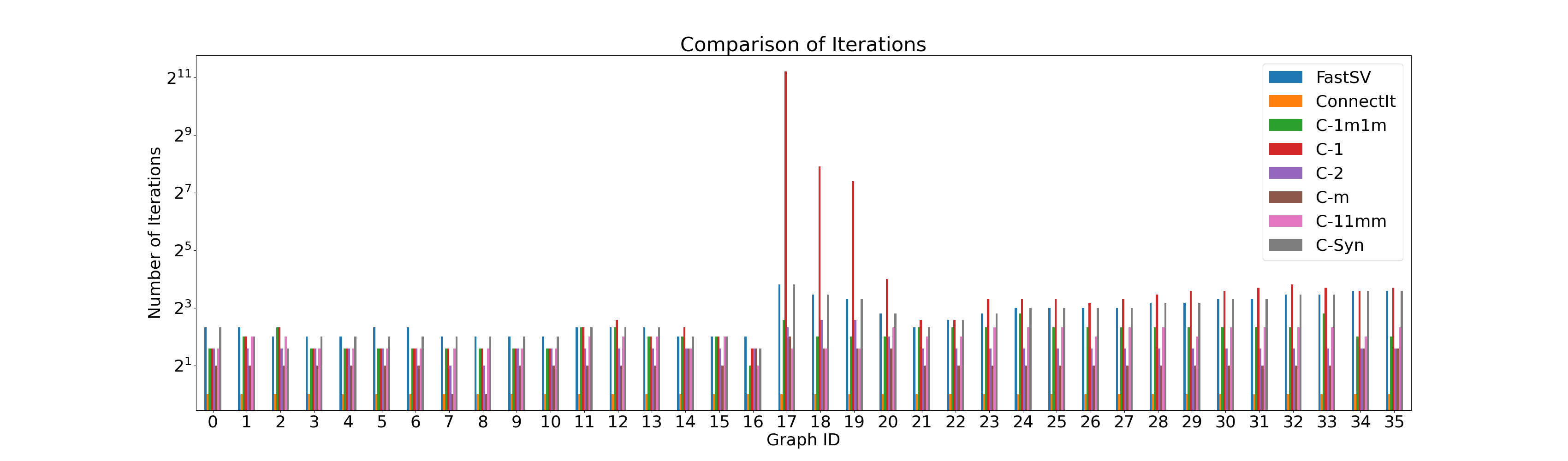}
    \caption{Number of Iterations of \emph{FastSV}, \emph{ConnectIt}, and Different \emph{Contour} Variants.}
    \label{fig:1LocaleIteNum}
\end{figure*}

In Fig. \ref{fig:1LocaleIteNum}, we observe that for different graphs, the \emph{C-1} operator consistently requires the largest total number of iterations. Notably, Graph 17 \emph{road\_usa} exhibits the highest iteration count at 2369 iterations. This behavior is expected as \emph{C-1} represents the lowest-order minimum mapping operator, only considering directly connected vertices or those within a search distance of 1.

Comparatively, \emph{C-2} performs significantly better than \emph{C-1} in iteration numbers, involving all vertices that are at a distance of 2 from each edge $e=(v,w)$. Consequently, even a minimal increase in the order of the minimum mapping operator leads to a significant reduction in the total number of iterations for long-diameter graphs. For instance, Graph 17 \emph{road\_usa} only requires 5 iterations when using \emph{C-2}.

Further increasing the minimum mapping order to \emph{C-m} (here $m=1024$) yields additional reductions in the total number of iterations, but the improvement is not as significant. Across all graphs, \emph{C-m} achieves a maximum reduction of 3 iterations compared to \emph{C-2}. Therefore, we observe the following relationship for the total number of iterations: \emph{Number of Iterations (C-m)} $\leq$ \emph{Number of Iterations (C-2)} $\leq$ \emph{Number of Iterations (C-1)}.

Next, we analyze the behaviors of the combined minimum mapping operators \emph{C-11mm} and \emph{1m1m}. Among the 38 graphs, the majority (21) exhibit the same number of iterations for both operators. For the remaining 13 graphs, \emph{1m1m} shows a slightly higher number of iterations than \emph{C-11mm}. Thus, \emph{C-11mm} generally demonstrates a slightly better performance than \emph{1m1m} in terms of iteration count. Additionally, \emph{C-11mm} exhibits a total iteration count that is close to \emph{C-2}.

Comparing \emph{C-Syn} with \emph{FastSV}, we find them to be quite similar in terms of the total number of iterations. However, \emph{C-Syn} possesses a more efficient and simplified minimum mapping operator, contributing to the slight advantage in iteration count over \emph{FastSV}. The optimized \emph{C-2} operator significantly reduces the number of iterations compared to \emph{C-Syn} in most cases, validating the effectiveness of our optimization in reducing iterations, as also reflected in subsection \ref{subsec:exetime}.

\emph{ConnectIt}, as a non-iteration-based method, requires one union operation on all edges and one compression operation on all vertices. Consequently, we assign the iteration count for \emph{ConnectIt} as 1 for all graphs.

In summary, the average number of iterations, from low to high, are as follows: \emph{C-m}=2.19, \emph{C-2}=3.19, \emph{C-11mm}=3.89, \emph{C-1m1m}=4.31, \emph{C-Syn}=6.83, \emph{FastSV}=6.97, \emph{C-1}=83.86.

\subsection{Execution Time} \label{subsec:exetime}
\begin{figure*}
    \centering
    \hspace*{-2.6cm}
    \includegraphics[width=1.29\linewidth]{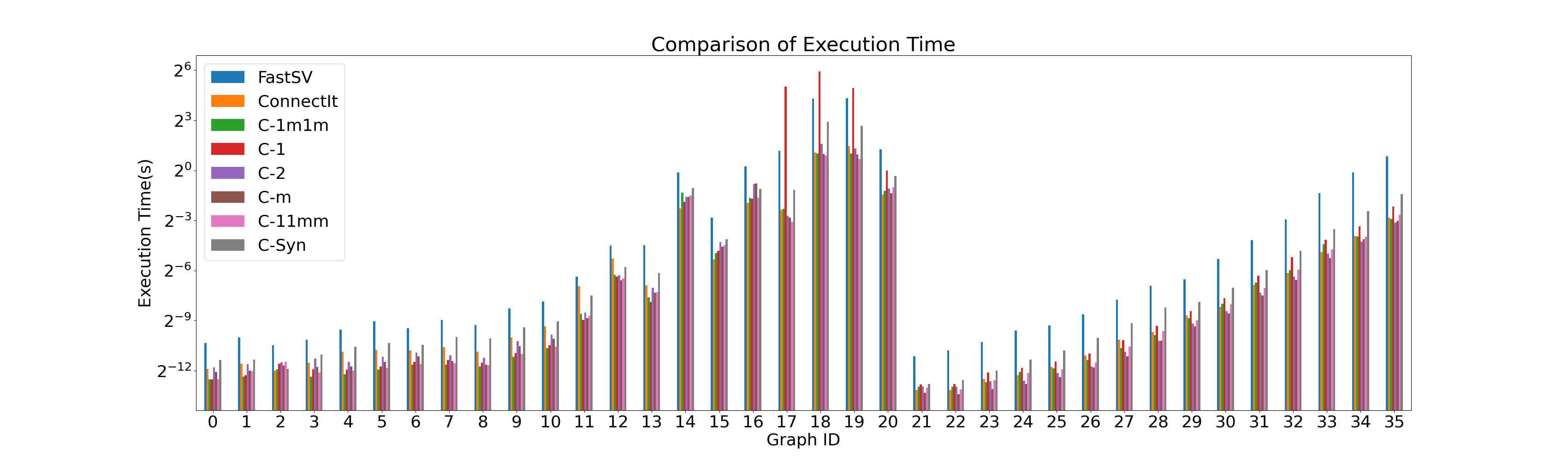}
    \caption{Execution Time of \emph{FastSV}, \emph{ConnectIt}, and Different \emph{Contour} Variants.}
    \label{fig:1LocaleExeTime}
    \vspace{-10.pt}
\end{figure*}
In Fig. \ref{fig:1LocaleExeTime}, we observe the execution times of different methods. Notably, there is a general trend that as the size of graphs increases (measured by the total number of edges and vertices), the execution time also increases. This pattern is expected since our server has a fixed number of 20 cores. As the graph size grows, each core has to handle a larger number of edges, leading to increased execution times. However, the execution times may vary due to differences in graph topology.

Analyzing the Delaunay graphs, which share similar topology, we find that as the graph size grows from \emph{delaunay\_n10} to \emph{delaunay\_n24} (both the number of edges and vertices increase about  16000 times), the execution time of \emph{C-2} increases by 895 times, \emph{C-1m1m} increases by 1072 times, \emph{C-m} increases by 1268 times, \emph{ConnectIt} increases by 1303 times, \emph{C-11mm} increases by 1329 times, \emph{C-Syn} increases by 2705 times, and \emph{FastSV} increases by 4096 times.

Additionally, we observe that, in most cases, \emph{FastSV} exhibits longer execution times compared to all other methods. Only when the diameters of some graphs are particularly large does the execution time of \emph{C-1} surpass \emph{FastSV}. Moreover, \emph{C-Syn} consistently shows longer execution times compared to other \emph{Contour} variants. As mentioned previously, this is due to \emph{C-Syn} employing synchronous updates instead of immediate asynchronous updates, which hinders the quick spreading of small labels to other vertices, thereby reducing its convergence speed.

In summary, the execution times of the algorithms generally follow the trend of increasing with graph size. However, specific algorithm characteristics, such as synchronous vs. asynchronous updates, also play a significant role in determining execution times. The overall performance of our \emph{Contour} algorithm outperforms \emph{FastSV}, highlighting the effectiveness of our optimization strategies in reducing execution times.

\subsection{Speedup compared with FastSV} \label{subsec:speedup1}
\begin{figure*}
    \centering
    \hspace*{-2.6cm}
    \includegraphics[width=1.29\linewidth]{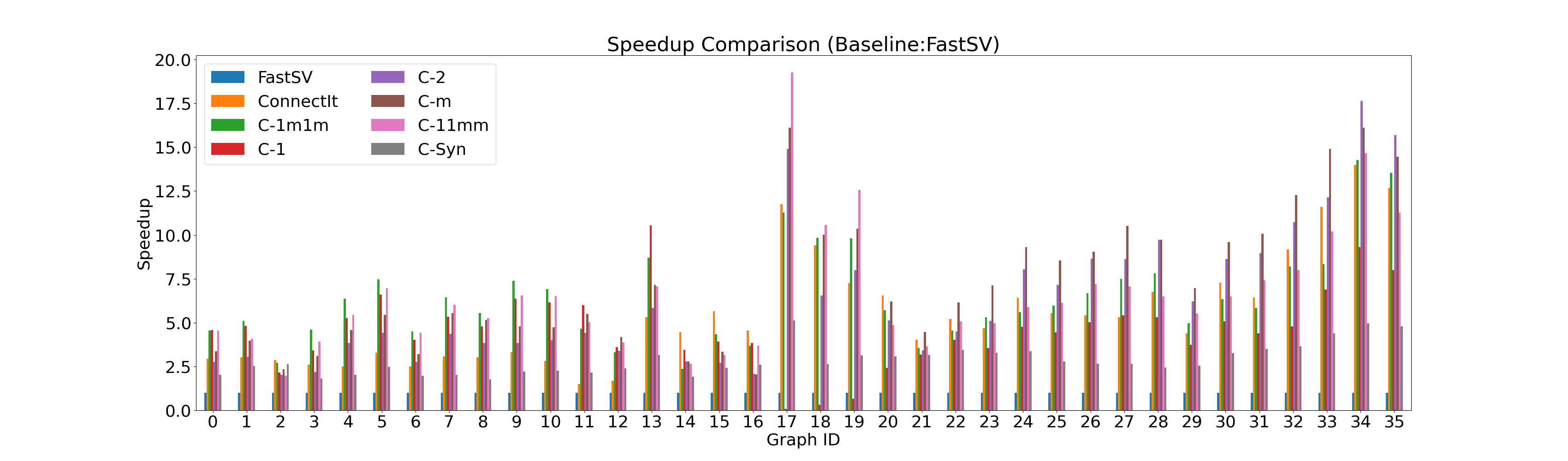}
    \caption{Speedups of \emph{ConnectIt} and Different \emph{Contour} Variants compared with \emph{FastSV}.}
    \label{fig:1LocaleSpeedup1}
    \vspace{-10.pt}
\end{figure*}

In Fig. \ref{fig:1LocaleSpeedup1}, we observe the speedups of all the methods compared to \emph{FastSV}. The average speedups, from highest to lowest, are as follows: \emph{C-m} with a speedup of 7.3, \emph{C-11mm} with 6.6, \emph{ConnectIt} with 6.49, \emph{C-1m1m} with 6.33, \emph{C-2} with 6.33, \emph{C-1} with 4.62, and \emph{C-Syn} with 2.87. This indicates that high-order minimum mapping operators often perform significantly better than \emph{FastSV}.

However, \emph{C-1} shows particularly good speedup when the size and diameter of the graphs are small. The reason behind this behavior is that when the diameter of a graph is small, \emph{C-1} can converge quickly within a few iterations. Additionally, for each iteration, the total workload for each core is very low because of the small graph size and \emph{C-1}'s focus on only checking one-path neighbors. This operation exhibits excellent locality, which can be explained clearly by the work-depth model \cite{blelloch2010parallel} well. Consequently, under these conditions, \emph{C-1} achieves better performance. However, for larger graphs or graphs with higher diameters, \emph{C-1} cannot maintain this better speedup compared to other variants due to a higher overall workload or larger number of iterations.

\emph{C-m} achieves the best average speedup, but it may not be suitable for all cases because each iteration will have a higher cost. As mentioned above, \emph{C-1} excels in scenarios with small diameters and sizes, while \emph{C-m} is most effective for large-diameter or large-size graphs. It reduces the total number of iterations to minimize the overall cost.

\emph{C-2}, on the other hand, exhibits a relatively small cost in each iteration as it only checks reachable vertices within two steps. Simultaneously, it can significantly reduce the total number of iterations for graphs with large diameters. Thus, \emph{C-2} stands as a stable and simple operator that fits well in most cases.

\emph{C-1m1m} is also a stable operator, but its policy differs from \emph{C-2}. It alternates between two extreme operators, \emph{C-1} and \emph{C-m}. \emph{C-1} reduces the cost of each iteration, while \emph{C-m} focuses on minimizing the total number of iterations. Combining these two operators optimizes the overall performance.

The strategy behind \emph{C-11mm} is different. It attempts to handle graphs with the smallest cost first. If, after several iterations, the graph does not converge, \emph{C-11mm} employs the \emph{C-m} operator to reduce the total number of iterations rapidly. When a graph contains both very small and very large diameter components, \emph{C-11mm} quickly converges the small diameter components with minimal cost before efficiently handling the large diameter components using the \emph{C-m} operator.

In summary, the speedup of the algorithms compared to \emph{FastSV} exhibits variations based on the size and characteristics of the graphs. Different operators are more suitable for different scenarios, depending on graph size, and diameter. The overall performance of our \emph{Contour} algorithm outperforms \emph{FastSV} in many cases, particularly when utilizing high-order minimum mapping operators, validating the effectiveness of our approach.

\subsection{Speedup compared with ConnectIt} \label{subsec:speedup2}
\begin{figure*}
    \centering
   \hspace*{-3cm}  
   \includegraphics[width=1.29\linewidth]{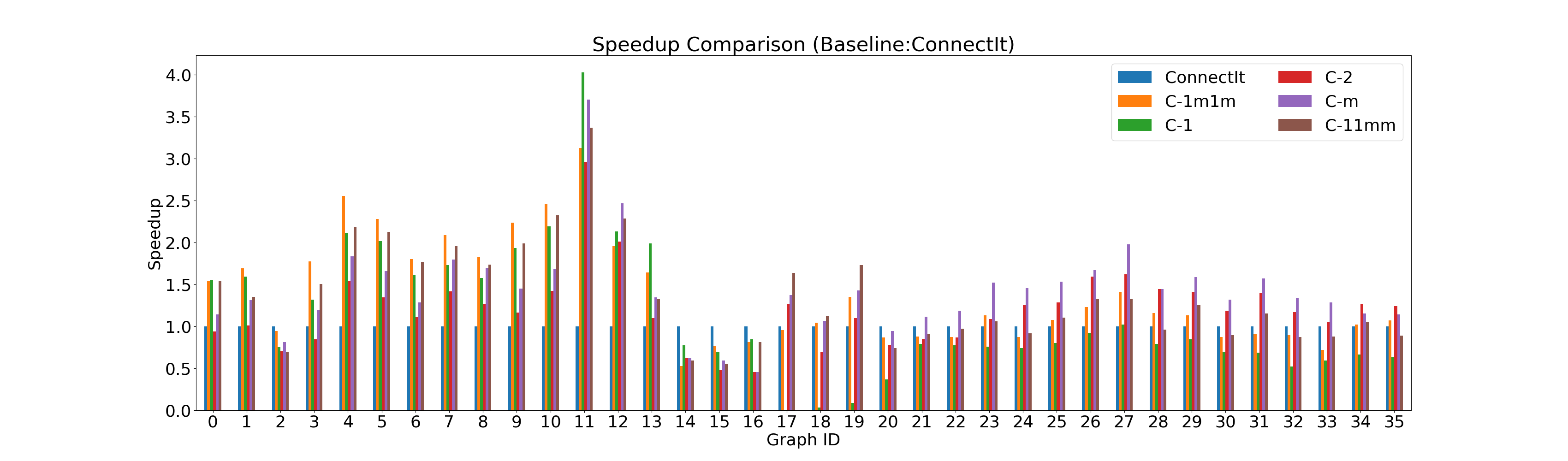}
    \caption{Speedups of Different \emph{Contour} Variants compared with \emph{ConnectIt}.}
    \vspace{-10.pt}
    \label{fig:1LocaleSpeedup2}
\end{figure*}
In Fig. \ref{fig:1LocaleSpeedup2}, we examine the speedups of our \emph{Contour} algorithm compared to another state-of-the-art algorithm, \emph{ConnectIt}. We will expose another perspective that can significantly affect the performance of different algorithms. 

Across the 36 graphs, \emph{C-m} outperforms \emph{ConnectIt} on 31 graphs, with an average speedup of 1.41. Similarly, \emph{C-2} achieves better performance on 26 graphs, with an average speedup of 1.2. Both \emph{C-1m1m} and \emph{C-11mm} outperform \emph{ConnectIt} on 23 graphs, with average speedups of 1.37 and 1.35, respectively. All of these \emph{Contour} variants achieve better performance on more than half of the graphs. \emph{C-1} shows better performance on 14 graphs, with an average speedup of 1.11. \emph{C-Syn}, on the other hand, only outperforms \emph{ConnectIt} on 2 graphs, with an average speedup of 0.62.

The experimental results provide valuable insights into when to use \emph{Contour} algorithms and when to use \emph{ConnectIt} to achieve better performance. In general, when we have a sufficient number of parallel cores to significantly reduce the cost of one iteration, employing our \emph{Contour} algorithm will lead to better performance. The \emph{Contour} algorithm's efficiency lies in its ability to reduce the total number of iterations and workload per iteration, resulting in overall speedup. However, if the graph size is very large, and the number of parallel cores is relatively small, each core will have to handle a considerable number of edges in each iteration, limiting the parallel effect. This is very similar to sequential instead of parallel computing. In such scenarios, the performance improvement is driven by high efficiency instead of high scalability because the system cannot provide sufficient parallel resources, where \emph{ConnectIt} excels with almost linear time complexity, approaching optimality. Thus, \emph{ConnectIt} can achieve better performance when the workload per core is significantly high or when the system lacks parallel resources. 

The work-depth model can clarify these results. When the work per iteration is high and parallel resources are limited, \emph{ConnectIt} stands as an ideal choice since it requires only one iteration. Conversely, when parallel resources can significantly reduce the work per iteration, \emph{Contour} algorithms achieve better overall performance with their ability to tolerate more iterations.

In conclusion, the choice between \emph{Contour} algorithms and \emph{ConnectIt} depends on the available parallel resources or the size of different graphs. Our \emph{Contour} algorithm demonstrates superior performance when enough parallel resources are available, but \emph{ConnectIt} remains a suitable choice for scenarios with high workloads and limited parallel resources.

\subsection{Distributed Memory Results}

The previous sections' results were based on shared memory parallel execution. However, when we consider distributed memory parallel executions involving multiple computing nodes, the absolute execution times become much longer. In practical scenarios, using multiple distributed memory computing nodes to solve a problem with a much longer time is not reasonable if it can be handled by a single shared memory parallel node with much less time. Therefore, we just give a brief summary instead of the detailed experimental results as follows.

When comparing with \emph{FastSV}, our \emph{Contour} algorithm demonstrates significantly better speedup than that in the shared memory parallel node setting. Among all the variants of our \emph{Contour} algorithm, \emph{C-1} achieves much better speedup when the total number of iterations is relatively low. The reason for this lies in \emph{C-1}'s ability to achieve high locality and reduce additional communication. Communication becomes a major performance bottleneck in distributed system scenarios, overshadowing computation.

Taking advantage of high-level parallel language Chapel, the shared memory \emph{ConnectIt} algorithm can be run on distributed systems. Similarly, \emph{ConnectIt} exhibits better relative performance compared to \emph{Contour} when dealing with large-sized graphs. Due to relatively less communication overhead, \emph{ConnectIt} even achieves good performance for middle-sized graphs. For small and low-diameter graphs, \emph{C-1}, \emph{C-11mm}, and \emph{C-1m1m} are more efficient and offer better performance.


\section{Related Work}
For connected component problems, graph traversal methods\cite{dhulipala2021theoretically,shun2013ligra,shun2014simple,slota2014bfs} 
have a major problem where they cannot achieve high performance when graph diameters are large or a graph has many small components. Label propagation methods \cite{esfahani2021thrifty,stergiou2018shortcutting,raghavan2007near,shun2013ligra,slota2014bfs} cannot converge fast when the graph diameter is large.

The Shiloach-Vishkin (SV) algorithm \cite{SHILOACH198257} is the pioneering tree-based hooking-compressing method to reduce the total number of iterations efficiently. There are different kinds of improvements to the SV algorithm. Awerbuch and Shiloach (AS) \cite{awerbuch1987new} use a very efficient parallelization using proper computational primitives and sparse data structures. The AS algorithm only keeps the information of the current forest and the convergence criterion for AS is to check whether each tree is a star. Afforest \cite{sutton2018optimizing} is an extension of the SV algorithm that approaches optimal work efficiency by processing subgraphs in each iteration. The LACC \cite{azad2019lacc} algorithm uses linear algebraic primitives to implement connected components and is based on the PRAM AS algorithm. FastSV \cite{zhang2020fastsv} further simplifies and optimizes LACC's tree hooking and compressing method to improve the performance. Iteration-based tree hooking-compressing methods exploit large-scale parallel resources to reduce the cost of each iteration and the total number of iterations.

Union-find-based algorithms \cite{dhulipala2020connectit,galil1991data,patwary2010experiments} take advantage of the disjoint set data structure to reduce the total operations in one iteration. Tree-based methods try to reduce the number of iterations, but disjoint set-based methods focus on reducing the total number of operations. So, tree-based methods are suitable for large-scale parallel execution but disjoint set-based methods are good for parallel resources limited scenarios.

There are some works combining different methods together to optimize the performance further. Slota \emph{et al.} \cite{slota2016case} developed a distributed memory multi-step method that combines parallel \emph{BFS} and label propagation technique. The ParConnect algorithm \cite{jain2017adaptive} is based on both the SV algorithm and parallel breadth-first search (\emph{BFS}). ConnectIt \cite{dhulipala2020connectit} provides a  framework to provide different sampling strategies and tree hooking and compression schemes. 

Recently, different optimization methods for connected component problems have been proposed. Thrifty Label Propagation (TLP) algorithm \cite{esfahani2021thrifty} uses the skewed degree distribution of real-world graphs to develop their optimized label propagation algorithm.  Sutton \emph{et al.} \cite{sutton2018optimizing} uses sampling to find the connected components on a subset of the edges, which can be used to reduce the number of edge inspections when running connectivity on the remaining edges. 
 
We formulate the connected components as a contour line discover problem and develop different minimum mapping operators for different scenarios. Our method is flexible and simple. It can achieve high performance in different scenarios. 

\section{Conclusion}
In this study, we addressed the fundamental graph problem of finding connected components using a novel method called ``minimum mapping.'' Our approach is characterized by its simplicity, flexibility, and efficient implementation, setting it apart from existing state-of-the-art methods.

We proved that our method achieves convergence in $\mathcal{O}(\log_2(d_{max}))$ time, where $d_{max}$ represents the largest diameter among all components in a graph. Experimental results also show that our algorithm can converge in a small number of iterations for different graphs.

Experimental results showed that our \emph{Contour} method significantly outperforms the state-of-the-art large-scale parallel \emph{FastSV} method. Additionally, our method complements the state-of-the-art shared memory parallel \emph{ConnectIt} method. Notably, we have successfully integrated our method and the state-of-the-art methods into an open-source graph package, Arachne. Arachne extends an open-source framework for Python users, enabling efficient large-scale graph analytics on supercomputers. This integration empowers high-level Python users to conduct large graph analytics efficiently, regardless of their familiarity with supercomputing and large data processing intricacies.

\section*{Acknowledgment}
We appreciate the help from the Arkouda and the Chapel community when we integrated the algorithms into Arkouda. This research was funded in part by NSF grant number CCF-2109988.

\section{References}
\bibliographystyle{plain}
\bibliography{Arkouda-chapel,suffixarray,graph,conn_comps,unionfind}

\end{document}